\documentclass[10pt, oneside, reqno]{article}
\usepackage[letterpaper, margin=1.5in]{geometry}
\usepackage{amsmath}
\usepackage{amsthm}
\usepackage{mathtools}
\usepackage{url}
\usepackage[utf8]{inputenc}
\usepackage{cite}
\usepackage{authblk}
\usepackage{hyperref}
\usepackage{tikz}
\usepackage{float}
\usepackage{subcaption}
\usepackage{xspace}
\usepackage{enumitem}
\usepackage{appendix}
\usepackage{tabularx}

\setlist[itemize]{noitemsep}
\setlist[enumerate]{noitemsep}

\usetikzlibrary{graphs, quotes, fit, shapes.geometric, calc}
\tikzset{
    a/.style={circle, draw=black, fill=red!20, minimum size=1mm, font=\small},
    i/.style={circle, draw=black, fill=cyan!20, minimum size=1mm, font=\small},
    o/.style={circle, draw=black, fill=green!20, minimum size=1mm, font=\small},
    v/.style={circle, draw=black, fill=white, minimum size=1mm, font=\small},
    e/.style={thick}
}

\newtheorem{theorem}{Theorem}[section]

\theoremstyle{definition}
\newtheorem{definition}{Definition}[section]

\newcommand*{\fl}{\ensuremath{\mathcal{L}}}
\newcommand*{\fe}{\ensuremath{\mathcal{E}}}
\newcommand{\ladt}[1]{\fl_{\text{#1}}}

\newcommand*{\lra}{\longrightarrow}
\newcommand*{\llr}{\longleftrightarrow}

\newcommand*{\rots}{\ensuremath{\text{ROT}_s}\xspace}
\newcommand{\knot}[1]{\ensuremath{#1\text{-NOT}}}

\title{A Systematic Study of Single-Anchor Logical Gadgets}
\author{Fikret H. Güngör}
\date{November 2025}

\begin{document}
\maketitle

\begin{abstract}
We present a systematic study of logical gadgets for 3-coloring under a single anchor constraint, where only one color representing logical falsehood is fixed to a vertex.
We introduce a framework of what we call ladgets (logical gadgets), graph gadgets that implement Boolean functions.
Then, we define a set of core gadgets, called primitives, which help identify and analyze the logical behavior of ladgets.
Next, we examine the structure of several standard ladgets and present several structural constraints for ladgets.
Through an exhaustive search of all non-isomorphic connected graphs up to 10 vertices, we verify all minimal constructions for standard ladgets. Notably, we identify exactly two non-isomorphic minimal XNOR ladgets in approximately 29 billion gadget configurations, highlighting the rarity of gadgets capable of expressing logical behavior.
We also present an embedding technique that embeds ladgets with less than 3 inputs from 3-coloring into k-coloring.
Our work shows how the single anchor constraint creates a fundamentally different framework from the two anchor gadgets used in SAT reductions.
\end{abstract}

\section{Introduction}
The 3-coloring problem is one of the most fundamental NP-complete problems, asking whether a graph can be colored with 3 colors. Reductions of this problem involve developing \emph{logical gadgets}, small graphs that simulate Boolean operations within some coloring constraints.

Traditional approaches for developing logical gadgets generally use two anchors, two vertices fixed to colors representing truth and falsehood in the logical sense.
While this approach results in smaller and easier gadgets to work with because of the symmetry, it involves setting too many color constraints therefore it feels unnatural from a graph-theoretic perspective.
This raises this natural question: \emph{What happens when we break the symmetry?}
\bigskip

In this paper, we investigate logical gadgets under a \emph{single anchor constraint}, fixing a single color representing falsehood to a distinguished \emph{anchor vertex}, while truth values represented by two remaining colors interchangeably. At first, this seems like a minor change, but the asymmetric constraints force us to develop fundamentally different primitives and gadgets.

\section{Preliminaries}
\subsection{3-Coloring and Gadgets}
A \emph{proper $n$-coloring} of a graph $G = (V, E)$ is a function $c: V \to \{0, 1, 2, ..., n-1\}$ such that for every edge $(u, v) \in E$, we have $c(u) \neq c(v)$. The 3-coloring problem asks whether a given graph admits a proper 3-coloring.

Having defined graph coloring, we asked a natural question: \emph{What happens when we fix some colors to certain vertices and observe how this constraint affects other vertices?} More precisely, suppose we pick some vertices as \emph{inputs} by fixing their colors to specific values, and then examine the possible colors of a designated \emph{output} vertex in any valid coloring scheme. This input-output behavior is similar to a machine, or a mathematical function: we provide inputs and observe the resulting output.
These graph structures, which encode input-output relationships through graph coloring, are called \emph{gadgets}. A gadget implements a function from input colors to output colors, constrained by the graph's internal structure.
\bigskip

Having established this framework, we ask: \emph{Can we encode Boolean operations within this gadget system?}

\subsection{Single anchor system}
The answer is yes. To implement Boolean functions, our gadgets must be able to distinguish between the values \textsc{true} and \textsc{false}. We achieve this by introducing \emph{anchor vertices}, vertices that are fixed to a color regardless of the input configuration, serving as reference values within this framework.
Traditional constructions employ \emph{two anchor vertices}: one fixed to a color representing \textsc{true}, another to a color representing \textsc{false}. This approach provides symmetric reference points, resulting in cleaner and smaller gadget configurations. However, in a theoretical sense, instead of two anchors, we actually need only one anchor to fix some logical value for distinguishing between \textsc{true} and \textsc{false}. We call this a \emph{single anchor system}.

\begin{definition}[Single Anchor System]\label{def:sa-system}
In the single anchor system for k-coloring:
\begin{enumerate}
    \item A distinguished \emph{anchor vertex} $a_0$ is fixed to color $0$, representing logical falsehood.
    \item The remaining colors $\{1, 2, \ldots, k-1\}$ all representing logical truth, used interchangeably.
    \item For any vertex $v$, we interpret:
    \begin{itemize}
        \item $v \equiv 1$ iff $c(v) \in \{1, 2, \ldots, k-1\}$
        \item $v \equiv 0$ iff $c(v) = 0$
    \end{itemize}
\end{enumerate}
\end{definition}
We will adopt the single anchor system for 3-coloring throughout paper unless stated otherwise.

This system is \emph{minimal} in the sense that it uses the fewest possible fixed colors while still being able to express logical operations. However, with this asymmetry and lack of constraint on colors, gadgets become more complex and fundamentally new structures emerge.

\subsection{Basic Definitions and Notation}
Throughout this paper, we will work with single anchor systems, using the notation of \ref{def:sa-system}, and let $S = \{0, 1, 2\}$ to denote the set of all colors in 3-coloring. Let $T = \{1, 2\} \subset S$ to denote the true colors in $S$.

We denote the behavior of a gadget as a mapping from input colors to output colors. If a gadget maps an input color $i$ to a output color $j$, we write $i \to j$. If the output can be one of several colors, we write $i \to \{j_1, j_2, \ldots\}$. When two colors can map to each other interchangeably, we write $i \leftrightarrow j$.

When expressing the colors of vertices, we treat gadgets in the expression as functions maping input colors from $S$ to sets of possible output colors according to their mapping. For example, a gadget $\mathcal{G}$ with $n$ inputs can be viewed as the function $\mathcal{G} : S^n \to \wp(S)$.

A gadget is called \emph{reversible} if it retains the same mapping when its input and output vertices are swapped.

\section{Primitives}
\subsection{Definition}
\emph{Primitives} are small gadgets that implement basic operations. These primitives serve as building blocks for constructing more complex logical machinery within our system.

\subsection{Core Primitives}
\subsubsection{MOV (Move)}
\begin{figure}[H]
	\centering
	\begin{tikzpicture}
		\node[i, label=left:$i$] (1) at (0,0) {};
		\node[v] (2) at (2,1) {};
		\node[v] (3) at (2,-1) {};
		\node[o, label=right:$\theta$] (4) at (4,0) {};
	
		\draw[e] (1) -- (2);
		\draw[e] (1) -- (3);
		\draw[e] (2) -- (4);
		\draw[e] (3) -- (4);
		\draw[e] (2) -- (3);
	\end{tikzpicture}

	\caption{MOV gadget structure}
	\label{fig:mov}
\end{figure}
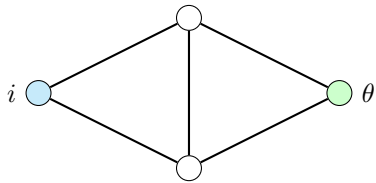

The MOV gadget shown above simply moves the input color to the output. One can verify the mapping, which is only:
\begin{align*}
	i \lra i
\end{align*}

\subsubsection{NOT}
\begin{figure}[H]
	\centering
	\begin{tikzpicture}
		\node[a, label=left:$a_0$] (1) at (-1,0) {};
		\node[v] (2) at (1,0) {};
		\node[i, label=left:$i$] (3) at (2,-1) {};
		\node[o, label=right:$\theta$] (4) at (2,1) {};
	
		\draw[e] (1) -- (2);
		\draw[e] (2) -- (3);
		\draw[e] (2) -- (4);
		\draw[e] (3) -- (4);
	\end{tikzpicture}

	\caption{NOT gadget structure}
	\label{fig:not}
\end{figure}
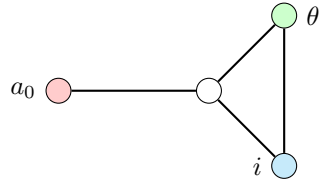

The NOT gadget shown above implements Boolean NOT operation and has the following mapping:
\begin{align*}
	0 &\llr \{1, 2\} \\
\end{align*}

Notably, this is the smallest gadget that's capable of implementing a Boolean operation.

\subsubsection{k-NOT}
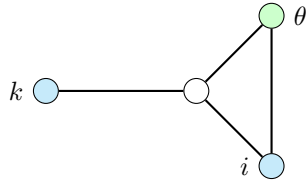
\begin{figure}[H]
	\centering
	\begin{tikzpicture}
		\node[i, label=left:$k$] (1) at (-1,0) {};
		\node[v] (2) at (1,0) {};
		\node[i, label=left:$i$] (3) at (2,-1) {};
		\node[o, label=right:$\theta$] (4) at (2,1) {};
	
		\draw[e] (1) -- (2);
		\draw[e] (2) -- (3);
		\draw[e] (2) -- (4);
		\draw[e] (3) -- (4);
	\end{tikzpicture}
	
	\caption{k-NOT gadget structure}
	\label{fig:k-not}
\end{figure}

After investigating some Boolean gates, we needed a primitive to express their behavior more concisely. Minimal gadgets have fewer vertices, which makes it easy for them to have some triangles. To address this, we generalize the NOT gadget, utilizing its triangle.
\bigskip

The \knot{k} gadget shown above has the following mapping:
\begin{align*}
	i = k &\lra S \setminus \{k\} \\
	\text{otherwise} &\lra k
\end{align*}

\subsubsection{ROT (Rotate)}
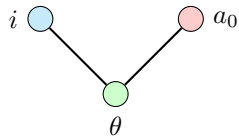
\begin{figure}[H]
	\centering
	\begin{tikzpicture}
		\node[i, label=left:$i$] (1) at (-1,1) {};
		\node[o, label=south:$\theta$] (2) at (0,0) {};
		\node[a, label=right:$a_0$] (3) at (1,1) {};
			
		\draw[e] (1) -- (2);
		\draw[e] (2) -- (3);
	\end{tikzpicture}

	\caption{ROT gadget structure}
	\label{fig:rot}
\end{figure}

The ROT gadget shown above has the following mapping:
\begin{align*}
	0 &\lra \{1, 2\} \\
	1 &\llr 2
\end{align*}
While mapping 0 to any true colors, this gadgets swaps the true colors between themselves.

\subsubsection{\rots (Rotate states 1 and 2)}
To take advantage of the asymmetric nature of the system, we tweak the idea of ROT, and develop a new gadget \rots. Unlike ROT, this gadget fixes 0 and swaps true colors between themselves.

\begin{figure}[H]
	\centering
	\begin{tikzpicture}
		\node[v] (1) at (0, 0) {};
		\node[v] (2) at (0, 2) {};
		\node[i, label=left:$i$] (3) at (-2, 0) {};
		\node[v] (4) at (-2, 2) {};
		\node[o, label=right:$\theta$] (5) at (1, 1) {};
		\node[v] (6) at (-1, 1) {};
		\node[a, label=left:$a_0$] (7) at (-1, 3) {};
		
		\draw[e] (1) -- (2);
		\draw[e] (1) -- (3);
		\draw[e] (1) -- (5);
		\draw[e] (1) -- (6);
		\draw[e] (2) -- (4);
		\draw[e] (2) -- (5);
		\draw[e] (2) -- (7);
		\draw[e] (3) -- (4);
		\draw[e] (3) -- (6);
		\draw[e] (7) -- (6);
		\draw[e] (7) -- (4);
	\end{tikzpicture}

	\caption{\rots gadget structure}
	\label{fig:rot-s}
\end{figure}

This gadget is also reversible, as one can verify the mapping after swapping the input with the output. This mapping also offers great flexibility when working with the asymmetric structure.
Let us investigate its structure by decomposing it into two \knot{k}s.

\begin{figure}[H]
	\centering
	\begin{subfigure}[t]{0.48\textwidth}
		\centering
		\begin{tikzpicture}
			\node[v, fill=blue!50] (1) at (0, 0) {};
			\node[v] (2) at (0, 2) {};
			\node[i, label=left:$i$] (3) at (-2, 0) {};
			\node[v] (4) at (-2, 2) {};
			\node[o, label=right:$\theta$] (5) at (1, 1) {};
			\node[v] (6) at (-1, 1) {};
			\node[a, label=above:$a_0$] (7) at (-1, 3) {};
			
			\draw[e] (1) -- (2);
			\draw[e] (1) -- (3);
			\draw[e] (1) -- (5);
			\draw[e] (1) -- (6);
			\draw[e] (2) -- (4);
			\draw[e] (2) -- (5);
			\draw[e] (2) -- (7);
			\draw[e] (3) -- (4);
			\draw[e] (3) -- (6);
			\draw[e] (7) -- (6);
			\draw[e] (7) -- (4);
			
			\draw[blue, thick] (0.3,-0.3) -- (0.3,0.3) -- (-0.7,1.3) -- (-0.7,3.3) -- (-1.3,3.3) -- (-1.3,1.3) -- (-2.3,0.3) -- (-2.3,-0.3) -- cycle;
			\draw[blue] (0.3, 0) node[right] {$\knot{0}\ i$};
		\end{tikzpicture}
	\end{subfigure}
	\hfill
	\begin{subfigure}[t]{0.48\textwidth}
		\centering
		\begin{tikzpicture}
			\node[v] (1) at (0, 0) {};
			\node[v, fill=blue!50] (2) at (0, 2) {};
			\node[i, label=left:$i$] (3) at (-2, 0) {};
			\node[v] (4) at (-2, 2) {};
			\node[o, label=right:$\theta$] (5) at (1, 1) {};
			\node[v] (6) at (-1, 1) {};
			\node[a, label=above:$a_0$] (7) at (-1, 3) {};
			
			\draw[e] (1) -- (2);
			\draw[e] (1) -- (3);
			\draw[e] (1) -- (5);
			\draw[e] (1) -- (6);
			\draw[e] (2) -- (4);
			\draw[e] (2) -- (5);
			\draw[e] (2) -- (7);
			\draw[e] (3) -- (4);
			\draw[e] (3) -- (6);
			\draw[e] (7) -- (6);
			\draw[e] (7) -- (4);
			
			\draw[blue, thick] (-1.7,-0.3) -- (-2.3,-0.3) -- (-2.3,2.3) -- (-1.3,3.3) -- (-0.7,3.3) -- (0.3,2.3) -- (0.3,1.7) -- (-1.7,1.7) -- cycle;
			\draw[blue] (0.3, 2) node[right] {$\knot{i}\ 0$};
		\end{tikzpicture}
	\end{subfigure}

\label{fig:rot-s-not}
\caption{Decomposition of \rots to two \knot{k}s}
\end{figure}

As shown above, the \rots can be decomposed into two \knot{k}s: $\knot{0}\ i$ and $\knot{i}\ 0$. Thus, these constraints are sufficient to fully describe the gadget, as it can be constructed completely from them:
\begin{enumerate}
	\item $k \in \knot{i}\ 0$
	\item $m \in \knot{0}\ i \setminus \{k\}$
	\item $\theta \in S \setminus \{k, m\}$
\end{enumerate}

\begin{theorem} \rots implements the following mapping:
	\begin{align*}
		0 &\lra 0 \\
		1 &\llr 2
	\end{align*}
\end{theorem}

\begin{proof}
	We consider cases $i = 0$ and $i \in T$.
	
	\textbf{Case 1:} $i = 0$.
	Then $k \in T$ and $m \in T \setminus \{k\} = \text{ROT}\ k$. Since $k \in T$, $k \notin \text{ROT}\ k \subseteq T$ hence $k \neq m$. Thus, this yields $\theta = 0$.
	
	\textbf{Case 2:} $i \in T$.
	Then $k \in \knot{i}\ 0 = \{i\}$ and $m \in \knot{0}\ i \setminus \{k\} = \{0\}$. Thus, $\theta \in S \setminus \{i, 0\}$ yields $\theta \in T \setminus \{i\} = \text{ROT}\ i$.
	
	In summary, $i = 0$ corresponds to $\theta = 0$, and $i \in T$ corresponds to $\theta \in \text{ROT}\ i$.
\end{proof}

\section{Ladgets}
\subsection{Definition}
A \emph{ladget} (a \emph{logical gadget}) is a gadget that implements a Boolean function. Intuitively, a ladget $\fl$ acts as a machine that takes Boolean input values $(i_1, i_2, \ldots, i_n)$ and outputs a Boolean value $\theta$ that is identical across all possible 3-colorings for the same input assignment.

\begin{definition}[Ladget] A ladget $\fl$ is defined as follows:
	
	Let $\fl = (G, a_0, I, \theta)$, where $G = (V,E)$ is the graph, $a_0 \in V$ is the anchor vertex, $I = (i_1, i_2, \ldots, i_n)$ is the ordered tuple of input vertices, and $\theta \in V$ is the output vertex.
	
	We say that $\fl$ \emph{implements a Boolean function} $\varphi_\fl: \{0,1\}^n \to \{0,1\}$ where $\varphi_\fl$ \emph{depends on} every input, if the graph satisfies:
	\begin{enumerate}
		\item \textbf{Universality:} For every assignment of colors to the input vertices, there is a valid 3-coloring of $\fl$.
		\item \textbf{Consistency:} For every valid 3-coloring of $\fl$, if the input vertices have Boolean values $(v_1, v_2, \ldots, v_n)$, then the output vertex $\theta$ satisfies:
		\[
			\theta \equiv \varphi_\fl(v_1, v_2, \ldots, v_n)
		\]
	\end{enumerate}
\end{definition}

We refer to $|V|$ as the order of a ladget $\fl$, denoted by $|\fl|$.

\begin{definition}[Minimality]
	A ladget $\fl$ is called \emph{minimal} if there exists no other ladget implementing $\varphi_\fl$ in strictly smaller order than $|\fl|$.
\end{definition}


Now let us investigate several minimal ladgets in detail.

\subsection{Minimal Ladgets for Standard Operations}

\subsubsection{NAND}
\begin{figure}[H]
	\centering
	\begin{tikzpicture}
		\node[a, label=left:$a_0$] (1) at (0,0) {};
		\node[v] (2) at (1,0) {};
		\node[i, label=above:$i_1$] (3) at (2,1) {};
		\node[v] (4) at (2,-1) {};
		\node[o, label=45:$\theta$] (5) at (3,1) {};
		\node[v] (6) at (3,-1) {};
		\node[i, label=right:$i_2$] (7) at (4,0) {};
		
		\draw[e] (1) -- (2);
		\draw[e] (2) -- (3);
		\draw[e] (2) -- (4);
		\draw[e] (3) -- (4);
		\draw[e] (3) -- (5);
		\draw[e] (4) -- (6);
		\draw[e] (5) -- (6);
		\draw[e] (7) -- (5);
		\draw[e] (7) -- (6);
	\end{tikzpicture}
	
	\caption{$\ladt{NAND}$ ladget structure}
	\label{fig:nand}
\end{figure}

The NAND ladget $\ladt{NAND}$ shown above, has 7 vertices and it is one of the two minimal NAND ladgets according to our exhaustive search (Section~\ref{ssec:search-results}).

\begin{figure}[H]
	\centering
	\begin{subfigure}[t]{0.48\textwidth}
		\centering
		\begin{tikzpicture}
			\node[a, label=left:$a_0$] (1) at (0,0) {};
			\node[v] (2) at (1,0) {};
			\node[i, label=above:$i_1$] (3) at (2,1) {};
			\node[v, fill=blue!40] (4) at (2,-1) {};
			\node[o, label=45:$\theta$] (5) at (3,1) {};
			\node[v] (6) at (3,-1) {};
			\node[i, label=right:$i_2$] (7) at (4,0) {};
			
			\draw[e] (1) -- (2);
			\draw[e] (2) -- (3);
			\draw[e] (2) -- (4);
			\draw[e] (3) -- (4);
			\draw[e] (3) -- (5);
			\draw[e] (4) -- (6);
			\draw[e] (5) -- (6);
			\draw[e] (7) -- (5);
			\draw[e] (7) -- (6);
			
			\draw[blue, thick] (-0.3, -0.3) -- (-0.3, 0.3) -- (0.7, 0.3) -- (1.7, 1.3) -- (2.3, 1.3) -- (2.3, -1.3) -- (1.7, -1.3) -- (0.7, -0.3) -- cycle;
			\draw[blue] (1.7, -1.3) node[left] {$\knot{0}\ i_1$};
		\end{tikzpicture}
	\end{subfigure}
	\begin{subfigure}[t]{0.48\textwidth}
		\centering
		\begin{tikzpicture}
			\node[a, label=left:$a_0$] (1) at (0,0) {};
			\node[v] (2) at (1,0) {};
			\node[i, label=above:$i_1$] (3) at (2,1) {};
			\node[v] (4) at (2,-1) {};
			\node[o, label=45:$\theta$, fill=blue!40] (5) at (3,1) {};
			\node[v] (6) at (3,-1) {};
			\node[i, label=right:$i_2$] (7) at (4,0) {};
			
			\draw[e] (1) -- (2);
			\draw[e] (2) -- (3);
			\draw[e] (2) -- (4);
			\draw[e] (3) -- (4);
			\draw[e] (3) -- (5);
			\draw[e] (4) -- (6);
			\draw[e] (5) -- (6);
			\draw[e] (7) -- (5);
			\draw[e] (7) -- (6);
			
			\draw[blue, thick] (1.7, -0.7) -- (1.7, -1.3) -- (3.3, -1.3) -- (4.3, -0.3) -- (4.3, 0.3) -- (3.3, 1.3) -- (2.7, 1.3) -- (2.7, -0.7) -- cycle;
			\draw[blue] (1.7, -1.3) node[left] {$\knot{(\knot{0}\ i_1)}\ i_2$};
		\end{tikzpicture}
	\end{subfigure}
	
	\caption{$\ladt{NAND}$ decomposition}
	\label{fig:nand-decom}
\end{figure}

Using this decomposition, we can express $\theta$ as:
\[
	\theta \in \knot{(\knot{0}\ i_1)}\ i_2 \setminus \{i_1\}
\]

\begin{theorem} $\ladt{NAND}$ implements the Boolean NAND operation:
	\begin{align*}
		(1, 1) &\lra 0 \\
		\text{otherwise} &\lra 1
	\end{align*}
\end{theorem}

\begin{proof}
	We consider cases $\theta \equiv 0$ and $\theta \equiv 1$.
	
	\textbf{Case 1:} $\theta \equiv 0$.
	This implies that $i_1 \in T$. Since $\knot{0}\ i_1 = \{0\}$, $\theta \in \{0\} = \knot{0}\ i_2$, implying $i_2 \in T$.
	Thus, this yields $i_1 \equiv i_2 \equiv 1$.
	
	\textbf{Case 2:} $\theta \equiv 1$.
	Then $\knot{(\knot{0}\ i_1)}\ i_2 \setminus \{i_1\} \subseteq T$. $i_1 \in T$ requires $i_2 = 0$ since $\knot{0}\ i_1$ would equal $\{0\}$.
	Similarly, $i_2 \in T$ requires $i_1 = 0$ since we need $\knot{0}\ i_1 \subseteq T$. Also $i_1 = i_2 = 0$ follows with $\theta \in T$.
	Thus, this covers all the remaining cases.
	
	In summary, $\theta \equiv 0$ corresponds to $i_1 \equiv i_2 \equiv 1$, and $\theta \equiv 1$ corresponds to $i_1 \equiv 0$ or $i_2 \equiv 0$.
\end{proof}

\subsubsection{OR}
\begin{figure}[H]
	\centering
	\begin{tikzpicture}
		\node[a, label=left:$a_0$] (1) at (0,0) {};
		\node[v] (2) at (1,0) {};
		\node[v] (3) at (2,1) {};
		\node[i, label=left:$i_1$] (4) at (2,-1) {};
		\node[v] (5) at (4,1) {};
		\node[v] (6) at (4,-1) {};
		\node[i, label=right:$i_2$] (7) at (5,0) {};
		\node[o, label=right:$\theta$] (8) at (3, 2) {};
		
		\draw[e] (1) -- (2);
		\draw[e] (2) -- (3);
		\draw[e] (2) -- (4);
		\draw[e] (3) -- (4);
		\draw[e] (3) -- (5);
		\draw[e] (4) -- (6);
		\draw[e] (5) -- (6);
		\draw[e] (7) -- (5);
		\draw[e] (7) -- (6);
		\draw[e] (8) -- (3);
		\draw[e] (8) -- (5);
	\end{tikzpicture}
	
	\caption{$\ladt{OR}$ ladget structure}
	\label{fig:or}
\end{figure}

The OR ladget $\ladt{OR}$ shown above has 8 vertices and it is one of the two minimal OR ladgets according to our exhaustive search (Section~\ref{ssec:search-results}). Interestingly, this gadget resembles the NAND gate we analyzed earlier, as it only has one additional vertex and two edges, allowing us to decompose similarly using the Boolean operation NAND when investigating this graph:

\begin{figure}[H]
	\centering
	\begin{subfigure}[t]{0.48\textwidth}
		\centering
		\begin{tikzpicture}
			\node[a, label=left:$a_0$] (1) at (0,0) {};
			\node[v] (2) at (1,0) {};
			\node[v, label=above:$m$, fill=blue!40] (3) at (2,1) {};
			\node[i, label=below:$i_1$] (4) at (2,-1) {};
			\node[v, label=above:$k$, fill=blue!40] (5) at (4,1) {};
			\node[v] (6) at (4,-1) {};
			\node[i, label=right:$i_2$] (7) at (5,0) {};
			\node[o, label=right:$\theta$] (8) at (3, 2) {};
			
			\draw[e] (1) -- (2);
			\draw[e] (2) -- (3);
			\draw[e] (2) -- (4);
			\draw[e] (3) -- (4);
			\draw[e] (3) -- (5);
			\draw[e] (4) -- (6);
			\draw[e] (5) -- (6);
			\draw[e] (7) -- (5);
			\draw[e] (7) -- (6);
			\draw[e] (8) -- (3);
			\draw[e] (8) -- (5);
			
			\draw[blue, thick] (-0.3, -0.3) -- (-0.3, 0.3) -- (0.7, 0.3) -- (1.7, 1.3) -- (2.3, 1.3) -- (2.3, -1.3) -- (1.7, -1.3) -- (0.7, -0.3) -- cycle;
			\draw[blue] (1.7, 1.3) node[left] {$\text{NOT}\ i_1$};
		\end{tikzpicture}
	\end{subfigure}
	\begin{subfigure}[t]{0.48\textwidth}
		\centering
		\begin{tikzpicture}
			\node[a, label=left:$a_0$] (1) at (0,0) {};
			\node[v] (2) at (1,0) {};
			\node[v, fill=blue!40, label=above:$m$] (3) at (2,1) {};
			\node[i, label=below:$i_1$] (4) at (2,-1) {};
			\node[v, label=above:$k$, fill=blue!40] (5) at (4,1) {};
			\node[v] (6) at (4,-1) {};
			\node[i, label=right:$i_2$] (7) at (5,0) {};
			\node[o, label=right:$\theta$] (8) at (3, 2) {};
			
			\draw[e] (1) -- (2);
			\draw[e] (2) -- (3);
			\draw[e] (2) -- (4);
			\draw[e] (3) -- (4);
			\draw[e] (3) -- (5);
			\draw[e] (4) -- (6);
			\draw[e] (5) -- (6);
			\draw[e] (7) -- (5);
			\draw[e] (7) -- (6);
			\draw[e] (8) -- (3);
			\draw[e] (8) -- (5);
			
			\draw[blue, thick] (-0.3, -0.3) -- (-0.3, 0.3) -- (0.7, 0.3) -- (1.7, 1.3) -- (4.3, 1.3) -- (5.3, 0.3) -- (5.3, -0.3) -- (4.3, -1.3) -- (1.7, -1.3) -- (0.7, -0.3) -- cycle;
			\draw[blue] (4.3, 1.3) node[right] {$(\text{NOT}\ i_1)\ \text{NAND}\ i_2$};
		\end{tikzpicture}
	\end{subfigure}
	
	\caption{$\ladt{OR}$ decomposition}
	\label{fig:or-decom}
\end{figure}

Using the decomposition and labeling above, these constraints are sufficient to fully describe the gadget, as it can be constructed completely from them:
\begin{enumerate}
	\item $m \in \text{NOT}\ i_1$
	\item $k \in m\ \text{NAND}\ i_2$
	\item $\theta \in S \setminus \{m, k\}$
\end{enumerate}

\begin{theorem} $\ladt{OR}$ implements the Boolean OR operation:
	\begin{align*}
		(0, 0) &\lra 0 \\
		\text{otherwise} &\lra 1
	\end{align*}
\end{theorem}

\begin{proof}
	We consider cases $i_2 \equiv 0$ and $i_2 \equiv 1$.
	
	\textbf{Case 1:} $i_2 \equiv 0$.
	Then we observe that the graph structure is the same as $\rots$, resulting $\theta \in \rots\ i_1$.
	Thus, this yields $\theta \equiv i_1$.
	
	\textbf{Case 2:} $i_2 \equiv 1$.
	Then $k \in m\ \text{NAND}\ i_2$, following logically with $k \equiv \text{NOT}\ m \equiv \text{NOT}\ \text{NOT}\ i_1 \equiv i_1$. Since $m \equiv \text{NOT}\ i_1$ and $k \equiv i_1$, therefore either $m \equiv 0$ or $k \equiv 0$.
	Thus implying $\theta \equiv 1$.
	
	In summary, $\theta \equiv 0$ corresponds to $i_1 \equiv i_2 \equiv 0$, and $\theta \equiv 1$ corresponds to $i_1 \equiv 1$ or $i_2 \equiv 1$.
\end{proof}

\subsubsection{AND}
\begin{figure}[H]
	\centering
	\begin{tikzpicture}
		\node[i, label=left:$i_1$] (1) at (0,1) {};
		\node[i, label=left:$i_2$] (2) at (0,-1) {};
		\node[v] (3) at (1,0) {};
		\node[v] (4) at (2,2) {};
		\node[v] (5) at (2,-2) {};
		\node[a, label=below:$a_0$] (6) at (3,-3) {};
		\node[v] (7) at (4,-2) {};
		\node[o, label=right:$\theta$] (8) at (5, 1) {};
		
		\draw[e] (1) -- (3);
		\draw[e] (1) -- (4);
		\draw[e] (2) -- (3);
		\draw[e] (2) -- (5);
		\draw[e] (4) -- (5);
		\draw[e] (4) -- (8);
		\draw[e] (5) -- (6);
		\draw[e] (7) -- (8);
		\draw[e] (7) -- (6);
		\draw[e] (7) -- (3);
		\draw[e] (8) -- (3);
	\end{tikzpicture}
	
	\caption{$\ladt{AND}$ ladget structure}
	\label{fig:and}
\end{figure}

The AND ladget $\ladt{AND}$ shown above has 8 vertices and it is one of the three minimal AND ladgets according to our exhaustive search (Section~\ref{ssec:search-results}).

\begin{figure}[H]
	\centering
	\begin{tikzpicture}
		\node[i, label=left:$i_1$] (1) at (0,1) {};
		\node[i, label=left:$i_2$] (2) at (0,-1) {};
		\node[v, label=left:$l$, fill=blue!40] (3) at (1,0) {};
		\node[v, label=above:$m$, fill=blue!40] (4) at (2,2) {};
		\node[v, label=below:$k$, fill=blue!40] (5) at (2,-2) {};
		\node[a, label=left:$a_0$] (6) at (3,-3) {};
		\node[v] (7) at (4,-2) {};
		\node[o, fill=blue!40] (8) at (5, 1) {};
		
		\draw[e] (1) -- (3);
		\draw[e] (1) -- (4);
		\draw[e] (2) -- (3);
		\draw[e] (2) -- (5);
		\draw[e] (4) -- (5);
		\draw[e] (4) -- (8);
		\draw[e] (5) -- (6);
		\draw[e] (7) -- (8);
		\draw[e] (7) -- (6);
		\draw[e] (7) -- (3);
		\draw[e] (8) -- (3);
		
		\draw[blue, thick] (2.7, -3.3) -- (3.3, -3.3) -- (4.3, -2.3) -- (5.3, 0.7) -- (5.3, 1.3) -- (4.7, 1.3) -- (0.7, 0.3) -- (0.7, -0.3) -- (3.3, -2) -- (2.7, -2.7) -- cycle;
		\draw[blue] (5.3, 1) node[right] {$\text{NOT}\ l$};
	\end{tikzpicture}
	
	\caption{$\ladt{AND}$ labeling and decomposition}
	\label{fig:and-decom}
\end{figure}
Using the decomposition and labeling above, these constraints are sufficient to fully describe the gadget, as it can be constructed completely from them:
\begin{enumerate}
	\item $l \in S \setminus \{i_1, i_2\}$
	\item $k \in \text{ROT}\ i_2$
	\item $m \in S \setminus \{i_1, k\}$
	\item $\theta \in (\text{NOT}\ l) \setminus \{m\}$
\end{enumerate}

\begin{theorem} $\ladt{AND}$ implements the Boolean AND operation:
	\begin{align*}
		(1, 1) &\lra 1 \\
		\text{otherwise} &\lra 0
	\end{align*}
\end{theorem}

\begin{proof}
	We consider cases $\theta \equiv 0$ and $\theta \equiv 1$.
	
	\textbf{Case 1:} $\theta \equiv 0$.
	Then $m \in T$ and $\text{NOT}\ l = \{0\}$, hence $l \in T$.
	Since $l \in S \setminus \{i_1, i_2\}$, this yields $i_1 \equiv 0$ or $i_2 \equiv 0$.
	
	\textbf{Case 2:} $\theta \equiv 1$.
	Then $(\text{NOT}\ l) \setminus \{m\} \subseteq T$ requires $l = 0$.
	Then $0 \in S \setminus \{i_1, i_2\}$, therefore $i_1 \not\equiv 0 \not\equiv i_2$.
	
	In summary, $\theta \equiv 0$ corresponds to $i_1 \equiv 0$ or $i_2 \equiv 0$, and $\theta \equiv 1$ corresponds to $i_1 \equiv i_2 \equiv 1$.
\end{proof} 

\subsubsection{XOR}
\begin{figure}[H]
	\centering
	\begin{tikzpicture}
		\node[a, label=right:$a_0$] (1) at (2,2) {};
		\node[v] (2) at (2,0) {};
		\node[v] (3) at (0,2) {};
		\node[o, label=right:$\theta$] (4) at (2,-2) {};
		\node[v] (5) at (1,-1) {};
		\node[v] (6) at (-1,1) {};
		\node[i, label=left:$i_1$] (7) at (-2,2) {};
		\node[v] (8) at (0, -2) {};
		\node[v] (9) at (-2, 0) {};
		\node[i, label=left:$i_2$] (10) at (-2, -2) {};
		
		\draw[e] (1) -- (2);
		\draw[e] (1) -- (3);
		\draw[e] (2) -- (3);
		\draw[e] (2) -- (4);
		\draw[e] (2) -- (5);
		\draw[e] (3) -- (6);
		\draw[e] (3) -- (7);
		\draw[e] (4) -- (5);
		\draw[e] (6) -- (7);
		\draw[e] (4) -- (8);
		\draw[e] (7) -- (9);
		\draw[e] (5) -- (9);
		\draw[e] (6) -- (8);
		\draw[e] (8) -- (10);
		\draw[e] (9) -- (10);
	\end{tikzpicture}
	
	\caption{$\ladt{XOR}$ ladget structure}
	\label{fig:xor}
\end{figure}

The XOR ladget $\ladt{XOR}$ shown above has 10 vertices and it is one of the four minimal XOR ladgets according to our exhaustive search (Section~\ref{ssec:search-results}).

\begin{figure}[H]
	\centering
	\begin{tikzpicture}
		\node[a, label=right:$a_0$] (1) at (2,2) {};
		\node[v, label=right:$\text{ROT }m$] (2) at (2,0) {};
		\node[v, label=above:$m$, fill=blue!40] (3) at (0,2) {};
		\node[o, label=right:$\theta$] (4) at (2,-2) {};
		\node[v] (5) at (1,-1) {};
		\node[v] (6) at (-1,1) {};
		\node[i, label=left:$i_1$] (7) at (-2,2) {};
		\node[v, label=below:$k$, fill=blue!40] (8) at (0, -2) {};
		\node[v, label=left:$l$, fill=blue!40] (9) at (-2, 0) {};
		\node[i, label=left:$i_2$] (10) at (-2, -2) {};
		
		\draw[e] (1) -- (2);
		\draw[e] (1) -- (3);
		\draw[e] (2) -- (3);
		\draw[e] (2) -- (4);
		\draw[e] (2) -- (5);
		\draw[e] (3) -- (6);
		\draw[e] (3) -- (7);
		\draw[e] (4) -- (5);
		\draw[e] (6) -- (7);
		\draw[e] (4) -- (8);
		\draw[e] (7) -- (9);
		\draw[e] (5) -- (9);
		\draw[e] (6) -- (8);
		\draw[e] (8) -- (10);
		\draw[e] (9) -- (10);
		
		\draw[blue, thick] (-2.3, 0.3) -- (-1.7, 0.3) -- (0.7, -0.7) -- (1.7, 0.3) -- (2.3, 0.3) -- (2.3, -2.3) -- (1.7, -2.3) -- (0.7, -1.3) -- (-1.7, -0.3) -- (-2.3, -0.3) -- cycle;
		
		\draw[blue, thick] (0.3, -2.3) -- (0.3, -1.7) -- (-0.7, 0.7) -- (0.3, 1.7) -- (0.3, 2.3) -- (-2.3, 2.3) -- (-2.3, 1.7) -- (-1.3, 0.7) -- (-0.3, -1.7) -- (-0.3, -2.3) -- cycle;
	\end{tikzpicture}
	
	\caption{$\ladt{XOR}$ labeling and decomposition}
	\label{fig:xor-decom}
\end{figure}
Using the decomposition and labeling above, these constraints are sufficient to fully describe the gadget, as it can be constructed completely from them:
\begin{enumerate}
	\item $k \neq i_2$
	\item $l \in S \setminus \{i_1, i_2\}$
	\item $m \in \knot{k}\ i_1 \setminus \{0\}$
	\item $\theta \in \knot{l}(\text{ROT }m) \setminus \{k\}$
\end{enumerate}

\begin{theorem} $\ladt{XOR}$ implements the Boolean XOR operation:
	\begin{align*}
		i_1 \equiv i_2 &\lra 0 \\
		i_1 \not\equiv i_2 &\lra 1
	\end{align*}
\end{theorem}

\begin{proof}
	We consider cases $\theta \equiv 0$ and $\theta \equiv 1$.
	
	\textbf{Case 1:} $\theta \equiv 0$.
	Then $\knot{l}(\text{ROT }m) \setminus \{k\} = \{0\}$. Hence $k \in T$ and either $l = 0$ or $l \in \text{ROT }m$.
	
	\begin{itemize}
		\item \textbf{Subcase 1.1:} $l = 0$.
		Then $i_1 \neq 0 \neq i_2$ therefore $i_1 \equiv i_2 \equiv 1$.
		
		\item \textbf{Subcase 1.2:} $l \in \text{ROT }m$.
		Then $i_1, i_2 \notin \text{ROT }m$.
		Also $k = m$ since $m \in T$ and $\theta = \knot{l}(\text{ROT }m) \setminus \{k\}$ should output only 0.
		Since $i_2 \notin \text{ROT }m$ and $m = k \neq i_2$, $i_2 = 0$. 
		Also $k = m = \knot{m}\ i_1 \setminus \{0\}$ implies $i_1 \neq m$, hence $i_1 = 0$.
		Thus, $i_1 \equiv i_2 \equiv 0$.
	\end{itemize}
	
	\textbf{Case 2:} $\theta \equiv 1$.
	$\knot{l}(\text{ROT }m) \setminus \{k\} \subseteq T$ implies $l \in T$ since $m \in T$. Since $l \in T$, either $l \in \text{ROT }m$ or $l = m$.

	\begin{itemize}
		\item \textbf{Subcase 2.1:} $l \in \text{ROT }m$.
		This implies $k = 0$ since $\knot{l}(\text{ROT }m) = \{0, m\}$ but $\theta \in T$.
		Therefore $m \in \knot{0}\ i_1 \setminus \{0\}$.
		$0 = k \neq i_2$ and $0 \notin \knot{0}\ i_1$ yields $i_1 \equiv 0$ and $i_2 \equiv 1$.
		
		\item \textbf{Subcase 2.2:} $l = m$.
		This implies $k \neq m$ since $\knot{m}(\text{ROT }m) \setminus \{m\}$ would be $\emptyset$.
		Also $i_1 \neq m \neq i_2$. $m \in \knot{k}\ i_1 \setminus \{0\}$ implies $i_1 = k$, since $i_1 \neq k$ implies $m = k$ but $m \neq k$.
		$k \neq m$ therefore $k = 0$ or $k \in \text{ROT }m$.
		Since $i_1 = k$ and $i_2 \neq k$, $k = 0$ yields $i_1 = 0$ with $i_2 \in T$ and $k \in \text{ROT }m$ yields $i_1 \in \text{ROT }m$ with $i_2 = 0$.
		Thus, $i_1 \equiv 0 \not\equiv i_2$ or $i_1 \not\equiv 1 \equiv i_2$.
	\end{itemize}
	
	In summary, $\theta \equiv 0$ corresponds to $i_1 \equiv i_2$,
	and $\theta \equiv 1$ corresponds to $i_1 \not\equiv i_2$.
\end{proof}

In addition, we observe that $k \neq l$ in every case above. Therefore, connecting the vertices $k$ and $l$ retains the universality, yielding another minimal ladget that implements XOR.

Also, we flip the $\knot{0}\ \theta$ as follows and obtain the ladget $\ladt{XNOR}$ implementing XNOR:
\begin{figure}[H]
	\centering
	\begin{tikzpicture}
		\node[a, label=right:$a_0$] (1) at (2,2) {};
		\node[v] (2) at (2,0) {};
		\node[v] (3) at (0,2) {};
		\node[v] (4) at (2,-2) {};
		\node[o, label=right:$\theta$] (5) at (1,-1) {};
		\node[v] (6) at (-1,1) {};
		\node[i, label=left:$i_1$] (7) at (-2,2) {};
		\node[v, label=below:$k$] (8) at (0, -2) {};
		\node[v, label=left:$l$] (9) at (-2, 0) {};
		\node[i, label=left:$i_2$] (10) at (-2, -2) {};
		
		\draw[e] (1) -- (2);
		\draw[e] (1) -- (3);
		\draw[e] (2) -- (3);
		\draw[e] (2) -- (4);
		\draw[e] (2) -- (5);
		\draw[e] (3) -- (6);
		\draw[e] (3) -- (7);
		\draw[e] (4) -- (5);
		\draw[e] (6) -- (7);
		\draw[e] (4) -- (8);
		\draw[e] (7) -- (9);
		\draw[e] (5) -- (9);
		\draw[e] (6) -- (8);
		\draw[e] (8) -- (10);
		\draw[e] (9) -- (10);
		
		\draw[blue, thick] (1.7, 2.3) -- (1.7, 0.3) -- (0.7, -0.7) -- (0.7, -1.3) -- (1.7, -2.3) -- (2.3, -2.3) -- (2.3, 2.3) -- cycle;
	\end{tikzpicture}
	
	\caption{Obtaining $\ladt{XNOR}$ by flipping the $\knot{0}\ \theta$ in $\ladt{XOR}$}
	\label{fig:xnor}
\end{figure}

Likewise, we can also obtain another ladget that implements XNOR by connecting $k$ and $l$, since $k \neq l$. These two ladgets are the only minimal ladgets that implement XNOR.

\subsection{Structural Constraints}
\label{sec:structural-constraints}
Through analysis of valid ladgets, we identified several necessary structural 
properties that any minimal ladget must satisfy. These constraints can be used to efficiently filter (Section~\ref{ssec:verification}) the ladget space while preserving the valid minimal ladgets.

Let $G = (V, E)$ be the graph of ladget $\fl$ with designated vertices: anchor $a_0$, inputs $I = \{i_1, i_2, \ldots, i_n\}$, and output $\theta$.
We call vertices in the set $V \setminus (I \cup \{a_0, \theta\})$ \emph{internal vertices}.

\begin{theorem}[Universality Constraints]
	Let $\fl = (G, a_0, I, \theta)$ be a ladget, then:
	\begin{enumerate}
		\item Any two inputs are non-adjacent
		\item $a_0$ is not adjacent to any of the inputs
		\item No internal vertex is adjacent to three vertices of $I \cup \{a_0\}$
	\end{enumerate}
\end{theorem}

\begin{proof} Any violation would contradict universality.
	\begin{enumerate}
		\item Let $i_j, i_k \in I$ be adjacent. Then any valid 3-coloring must satisfy $c(i_j) \neq c(i_k)$. This forbids the input assignment $i_j = i_k$, violating universality.
		
		\item Let $i_j \in I$ and $a_0$ be adjacent. Then any valid 3-coloring must satisfy $c(i_j) \neq 0$. This forbids the input assignment $i_j = 0$, violating universality.
		
		\item Let an internal vertex $v$ be adjacent to three distinct vertices $k, l, m \in I \cup \{a_0\}$, then $c(v)$ must be different than $c(k), c(l), c(m)$.
		\begin{itemize}
			\item If one of $k, l, m$ is $a_0$, let $k = a_0$, then the assignment $l = 1, m = 2$ leaves no available color for $v$, violating universality.
			\item If all three are inputs, the assignment $k = 0, l = 1, m = 2$ leaves no available color for $v$, violating universality.
		\end{itemize}
	\end{enumerate}
\end{proof}

\begin{theorem}[Non-trivial Output]
	\label{thm:nontrivial-output}
	Let $\fl = (G, a_0, I, \theta)$ be a ladget, then $a_0$ and $\theta$ are non-adjacent.
\end{theorem}

\begin{proof}
	Let $\theta$ and $a_0$ be adjacent. Then any valid 3-coloring must satisfy $c(\theta) \neq c(a_0) = 0$.
	Thus, regardless of the input assignment, $\theta \equiv 1$, and $\fl$ implements a constant Boolean function $\varphi_\fl(i_1, \ldots, i_n) \equiv 1$, independent of any input.
\end{proof}

\begin{theorem}[Degree Constraints]
	Let $\fl = (G, a_0, I, \theta)$ be a ladget, then:
	\begin{enumerate}
		\item $\deg(\theta) \geq 2$
		\item If $\fl$ is minimal, then for any internal vertex $v$, $\deg(v) \geq 3$
	\end{enumerate}
\end{theorem}

\begin{proof} We consider the two conditions separately.
	\begin{enumerate}
		\item Suppose that $\deg(\theta) = 1$ and let $v \in V$ be the only vertex adjacent to $\theta$.
		Then $c(\theta) \neq c(v)$.
		\begin{itemize}
			\item If $c(v) \in \{1, 2\}$, then $\theta$ could take more than one color not responding to a single logical value, violating the definition.
			\item If $c(v) = 0$, then $\theta \equiv 1$ regardless of the input assignment, contradicting the non-trivial output requirement. (Theorem~\ref{thm:nontrivial-output})
		\end{itemize}
		
		\item Suppose that $\fl$ is minimal and let $v$ be an internal vertex with $\deg(v) \leq 2$.
		Then $v$ has at most 2 neighbors, therefore it is possible to always find a color assignment for $v$, not enforcing any rules between the adjacent vertices.
		Thus, $v$ is redundant and can be removed to obtain a smaller ladget $\fl'$, contradicting the minimality of $\fl$.
	\end{enumerate}
\end{proof}

\begin{theorem}[Input Degree]
	Let $\fl = (G, a_0, I, \theta)$ be a ladget, then for all inputs $i_j \in I$, $\deg(i_j) \geq 2$.
\end{theorem}

\begin{proof}
	For an input $i_j \in I$, suppose that $\deg(i_j) = 1$ and $\varphi_\fl$ depends on $i_j$. 
	Then $v \in V$ is the only vertex adjacent to $i_j$. Since $i_j$ is adjacent to only $v$, the color of $i_j$ propagates through the graph with only $v$, therefore there are no other vertices that has any information of $i_j$'s color that could constrain $v$ based on $i_j$.
	Consider two different colors $c_1, c_2 \in S$:
	\begin{itemize}
		\item If $i_j = c_1$, then $v \in S \setminus \{c_1\}$
		\item If $i_j = c_2$, then $v \in S \setminus \{c_2\}$
	\end{itemize}
	Since $(S \setminus \{c_1\}) \cap (S \setminus \{c_2\}) \neq \emptyset$, there exist one color $q_1$ in the intersection that $v$ can take under both input assignments.
	Even if $v$ is forced to take a different color, there exist a different color $q_2$ under a different input assignment since $v$ can't be constrained with 2 different colors as this would fix $v$ to a color regardless of $i_j$, making $\varphi_\fl$ independent of $i_j$.
	Thus, $v = q_1$ is possible in 2 input assignments, making it indistinguishable between 2 different inputs, rendering $\varphi_\fl$ independent of $i_j$. This contradiction shows that $\deg(i_j) \ge 2$.
\end{proof}

\section{Results} \label{sec:results}
\subsection{Exhaustive Search}
We performed exhaustive enumeration, generating all non-isomorphic connected graphs with 7-10 vertices using the \texttt{geng} tool from the \texttt{nauty} package \cite{nauty}, identifying all minimal ladgets implementing NAND, AND, OR, NOR, XOR, and XNOR.

\subsubsection{Verification} \label{ssec:verification}
For a 2-input gadget with $n$ vertices, each configuration includes:
\begin{itemize}
	\item 1 anchor vertex
	\item 2 input vertices
	\item 1 output vertex
	\item $n-4$ internal vertices
\end{itemize}

For each configuration $(a_0, i_1, i_2, \theta)$, we verified whether the graph implements a Boolean function as follows:
\begin{enumerate}
	\item \textbf{Structural check:} Verify that the gadget configuration meets the structural requirements presented in Section~\ref{sec:structural-constraints}.
	
	\item \textbf{Universality check:} For each input assignment $(i_1, i_2) \in \{0, 1, 2\}^2$, verify that a valid 3-coloring exists.
	
	\item \textbf{Consistency check:} For all valid colorings with the same input assignment, verify that the output $\theta$ has the same logical value.
	
	\item \textbf{Truth table comparison:} Compare the resulting truth table with the target Boolean function.
\end{enumerate}

In the structural check step, we were able to filter out approximately 98.8\% of the gadget space with 10 vertices, reducing the overall computation time by roughly 8 times.

\subsubsection{Implementation}
The search was implemented in Python using NetworkX \cite{networkx} with nauty's geng tool.

The verification was parallelized across 19 CPU cores, and the most computationally intensive search (10 vertices, approximately 29.5 billion configurations) required approximately 1 hour (previously 8 hours without structural checks) on Intel Core i9-13900H processor with 32 GB of memory.

\subsection{Search Results} \label{ssec:search-results}
\begin{table}[H]
	\centering
	\begin{tabular}{|l|c|c|c|c|c|}
		\hline
		Function & Vertices & Count & Rarity (in Graphs) & Rarity (in Configs) & Rarity (filtered) \\
		\hline
		NOT  & 4  & 1  & 1 in 6                & 1 in 72                   & 1 in 2              \\
		NAND & 7  & 2  & 1 in 426              & 1 in $1.79\times10^5$     & 1 in 622            \\
		AND  & 8  & 3  & 1 in 3{,}705          & 1 in $3.1\times10^6$      & 1 in 19{,}417       \\
		OR   & 8  & 2  & 1 in 5{,}558          & 1 in $4.6\times10^6$      & 1 in 29{,}126       \\
		NOR  & 10 & 20 & 1 in $5.85\times10^5$ & 1 in $1.4\times10^9$      & 1 in $1.75\times10^7$ \\
		XOR  & 10 & 4  & 1 in $2.9\times10^6$  & 1 in $7.3\times10^9$      & 1 in $8.76\times10^7$ \\
		XNOR & 10 & 2  & 1 in $5.8\times10^6$  & 1 in $1.47\times10^{10}$  & 1 in $1.75\times10^8$ \\
		\hline
	\end{tabular}
	\caption{Minimal implementations found through exhaustive search}
	\label{tab:results}
\end{table}

The full list of minimal implementations are available with graph6 strings in Appendix~\ref{app:ladgetlist}.
\bigskip

\textbf{Note:} While geng ensures graphs are non-isomorphic, different gadget configurations within the same graph may be isomorphic to each other under graph automorphisms. For instance, if a graph has reflectional symmetry (as in the case of $\ladt{XOR}$), two configurations that are reflections of each other implement the same function. The counts reported include only non-isomorphic gadgets.

\subsection{Rarity}
The exhaustive search reveals that ladgets are exceptionally rare in the space of all graphs. As shown in Table~\ref{tab:results}, the probability of a random graph configuration implementing a specific Boolean function decreases dramatically with gate complexity.

The most striking case is XNOR: of the approximately 29 billion 10-vertex configurations examined, only two non-isomorphic implementations exist. This corresponds to a rarity of 1 in 14.7 billion configurations (1 in 175.2 million after filtering out gadgets not meeting the structural requirements), or 1 in 5.8 million graphs. This extreme rarity demonstrates that logical expressiveness imposes harsh constraints on gadgets.


\section{Embedding into $k$-Coloring}
We present a simple technique to embed any 3-coloring ladget $\fl$ with $|I| \leq 2$ into $k$-coloring while preserving its Boolean function $\varphi_\fl$.

Let $\fl = (G, a_0, I, \theta)$ be a 3-coloring ladget with $|I| \leq 2$ and graph $G = (V, E)$.
We construct an \emph{embedding package} $\fe$, defined as the complete graph $K_{k - 3}$.
We then connect every vertex of $\fe$ to every vertex of $G$, obtaining a new graph $G'$ and then define a new $k$-coloring ladget $\fl_k = (G', a_0, I, \theta)$.
In any valid $k$-coloring, the vertices of $\fe$ occupy $k - 3$ distinct colors different from the anchor and input vertices, since all vertices in $\fe$ are adjacent to all vertices in $I \cup \{a_0\}$; these colors are considered redundant.
Since every vertex of $G$ is adjacent to all vertices of $\fe$, the additional $k - 3$ redundant colors are forbidden on $V$, restricting $G$ to a 3-coloring palette within the $k$-coloring scheme.

Thus, this construction embeds any 3-coloring ladget with less than 3 inputs into a valid $k$-coloring ladget without altering its logical behavior.
\bigskip

Below is an example embedding of $\ladt{NOT}$ into 5-coloring:
\begin{figure}[H]
	\centering
	\begin{subfigure}[t]{0.48\textwidth}
		\centering
		\begin{tikzpicture}
			\node[a, label=left:$a_0$] (1) at (-0.5,0.4) {};
			\node[v] (2) at (2,0.8) {};
			\node[i, label=right:$i$] (3) at (2.6,-1) {};
			\node[o, label=right:$\theta$] (4) at (3,0) {};
			
			\node[v] (5) at (1,0) {};
			\node[v] (6) at (1.4,-1) {};
			
			\draw[e] (1) -- (2);
			\draw[e] (2) -- (3);
			\draw[e] (2) -- (4);
			\draw[e] (3) -- (4);
			
			\draw[e] (5) -- (6);
			
			\foreach \x in {1, 2, 3, 4}{
				\draw[e, dotted] (\x) -- (5);
				\draw[e, dotted] (\x) -- (6);
			}
			
			\draw[blue, thick] (0.6, 0.3) -- (1.3, 0.3) -- (1.8,-1.3) -- (1.1, -1.3) -- cycle;
			\draw[blue] (1, -1.1) node[left] {$\fe$};
			
		\end{tikzpicture}
		\caption{Embedding package $\fe$ and $\ladt{NOT}$}
	\end{subfigure}
	\begin{subfigure}[t]{0.48\textwidth}
		\centering
		\begin{tikzpicture}
			\node[a, label=left:$a_0$] (1) at (-0.5,0.4) {};
			\node[v] (2) at (2,0.8) {};
			\node[i, label=right:$i$] (3) at (2.6,-1) {};
			\node[o, label=right:$\theta$] (4) at (3,0) {};
			
			\node[v] (5) at (1,0) {};
			\node[v] (6) at (1.4,-1) {};
			
			\draw[e] (1) -- (2);
			\draw[e] (2) -- (3);
			\draw[e] (2) -- (4);
			\draw[e] (3) -- (4);
			
			\draw[e] (5) -- (6);
			
			\foreach \x in {1, 2, 3, 4}{
				\draw[e] (\x) -- (5);
				\draw[e] (\x) -- (6);
			}
		\end{tikzpicture}
		\caption{After embedding}
	\end{subfigure}
	
	\caption{Embedding of $\ladt{NOT}$ into 5-coloring}
	\label{fig:not-4col}
\end{figure}


\bibliographystyle{plain}
\bibliography{refs}

@article{nauty,
	title = {Practical graph isomorphism, II},
	journal = {Journal of Symbolic Computation},
	volume = {60},
	pages = {94-112},
	year = {2014},
	issn = {0747-7171},
	doi = {https://doi.org/10.1016/j.jsc.2013.09.003},
	url = {https://www.sciencedirect.com/science/article/pii/S0747717113001193},
	author = {Brendan D. McKay and Adolfo Piperno},
}

@InProceedings{networkx,
	author =       {Aric A. Hagberg and Daniel A. Schult and Pieter J. Swart},
	title =        {Exploring Network Structure, Dynamics, and Function using NetworkX},
	booktitle =   {Proceedings of the 7th Python in Science Conference},
	pages =     {11 - 15},
	address = {Pasadena, CA USA},
	year =      {2008},
	editor =    {Ga\"el Varoquaux and Travis Vaught and Jarrod Millman},
}

\textit{Author e-mail address:}
\texttt{fikretgungor@hacettepe.edu.tr}

\clearpage
\appendixpage
\appendix

\section{List of Minimal Ladget Configurations} \label{app:ladgetlist}
Below is the list of minimal 3-coloring ladgets identified through our exhaustive search.

\begin{table}[H]
	\centering
	\caption*{Minimal NAND implementations}
	\begin{tabular}{|l|c|c|c|c|}
		\hline
		Graph6 & $a_0$ & $\theta$ & $i_1$ & $i_2$ \\
		\hline
		\verb"FCZeO" & 3 & 4 & 2 & 6 \\
		\verb"FCZUO" & 2 & 5 & 1 & 3 \\
		\hline
	\end{tabular}
\end{table}

\begin{table}[H]
	\centering
	\caption*{Minimal OR implementations}
	\begin{tabular}{|l|c|c|c|c|}
		\hline
		Graph6 & $a_0$ & $\theta$ & $i_1$ & $i_2$ \\
		\hline
		\verb"GCQeMo" & 2 & 3 & 4 & 6 \\
		\verb"G?optW" & 1 & 2 & 0 & 3 \\
		\hline
	\end{tabular}
\end{table}

\begin{table}[H]
	\centering
	\caption*{Minimal AND implementations}
	\begin{tabular}{|l|c|c|c|c|}
		\hline
		Graph6 & $a_0$ & $\theta$ & $i_1$ & $i_2$ \\
		\hline
		\verb"GCQbeK" & 4 & 7 & 2 & 3 \\
		\verb"G?q`ug" & 2 & 0 & 1 & 3 \\
		\verb"GCR`qk" & 0 & 7 & 2 & 4 \\
		\hline
	\end{tabular}
\end{table}

\begin{table}[H]
	\centering
	\caption*{Minimal NOR implementations}
	\begin{tabular}{|l|c|c|c|c|}
		\hline
		Graph6 & $a_0$ & $\theta$ & $i_1$ & $i_2$ \\
		\hline
		\verb"I?BD?psjO" & 3 & 4 & 1 & 6 \\
		\verb"I?BD?p{lO" & 4 & 3 & 1 & 6 \\
		\verb"I?BDAo{lO" & 4 & 3 & 1 & 6 \\
		\verb"I?BDAqtN_" & 4 & 3 & 1 & 6 \\
		\verb"I?B@`YYZ_" & 0 & 1 & 3 & 7 \\
		\verb"I?B@`ZYj_" & 1 & 0 & 3 & 7 \\
		\verb"I?`DAbkdo" & 1 & 4 & 2 & 3 \\
		\verb"I?`FAqkF_" & 2 & 0 & 3 & 5 \\
		\verb"I?`D`pesO" & 2 & 3 & 4 & 5 \\
		\verb"I?bBD`[s_" & 4 & 5 & 2 & 3 \\
		\verb"I?bB@qqQo" & 3 & 2 & 4 & 5 \\
		\verb"I?b@baidO" & 2 & 3 & 1 & 4 \\
		\verb"I?b@aTwy_" & 3 & 2 & 5 & 7 \\
		\verb"I?b@dpMh_" & 3 & 2 & 1 & 4 \\
		\verb"I?bERGwdG" & 2 & 3 & 1 & 4 \\
		\verb"I?`aeIqq_" & 2 & 3 & 4 & 7 \\
		\verb"I?`adIWoo" & 3 & 4 & 1 & 2 \\
		\verb"I?`eKpocg" & 2 & 1 & 3 & 4 \\
		\verb"I?`eHrWgo" & 3 & 4 & 1 & 2 \\
		\verb"I?`cmPogg" & 1 & 2 & 3 & 4 \\
		\hline
	\end{tabular}
\end{table}

\begin{table}[H]
	\centering
	\caption*{Minimal XOR implementations}
	\begin{tabular}{|l|c|c|c|c|}
		\hline
		Graph6 & $a_0$ & $\theta$ & $i_1$ & $i_2$ \\
		\hline
		\verb"I?`DU_[X_" & 2 & 5 & 0 & 3 \\
		\verb"I?`DU`eFO" & 2 & 5 & 0 & 3 \\
		\verb"I?b@dHYy?" & 3 & 2 & 4 & 5 \\
		\verb"I?b@bFWd_" & 2 & 1 & 4 & 9 \\
		\hline
	\end{tabular}
\end{table}

\begin{table}[H]
	\centering
	\caption*{Minimal XNOR implementations}
	\begin{tabular}{|l|c|c|c|c|}
		\hline
		Graph6 & $a_0$ & $\theta$ & $i_1$ & $i_2$ \\
		\hline
		\verb"I?`DU_[X_" & 2 & 1 & 0 & 3 \\
		\verb"I?`DU`eFO" & 2 & 1 & 0 & 3 \\
		\hline
	\end{tabular}
\end{table}

\end{document}